\documentclass[10pt,titlepage]{article}

\usepackage{graphicx}
\usepackage{endfloat}
\usepackage{amsfonts}
\usepackage{mychicago}
\usepackage{subfigure}
\usepackage{genetics_manu_style}

\usepackage{amsmath}
\usepackage{amssymb}
\usepackage{amsthm}
\usepackage{graphicx}
\usepackage{amsfonts}
\usepackage{subfigure}

\newtheorem{theorem}{Theorem}

\newtheorem{lemma}{Lemma}

\newcommand*\samethanks[1][\value{footnote}]{\footnotemark[#1]}
\author{Foo\thanks{University of Podunk, Timbuktoo}
\and Bar\samethanks
\and Baz\thanks{Somewhere Else}
\and Bof\samethanks[1]
\and Hmm\samethanks}

\title{The equilibrium allele frequency distribution 
for a population with reproductive skew
}

\author{Ricky Der\thanks{Department of Biology, University
of Pennsylvania, Philadelphia, PA 19104},
Joshua B. Plotkin\samethanks[1]}

\renewcommand{\CorrespondingAddress}{Department of Biology \\
    University of Pennsylvania \\
    Philadelphia, PA 19103 \\
    \texttt{rickyder@sas.upenn.edu} \vfill}
    \renewcommand{\RunningHead}{Equillibrium allele frequencies with reproductive skew}
\renewcommand{\CorrespondingAuthor}{Ricky Der}
\renewcommand{\KeyWords}{Wright-Fisher model, $\Lambda$-Fleming-Viot process, stationary distribution}

\begin{document}

\maketitle

\begin{abstract} 

We study the population genetics of two neutral alleles under reversible
mutation in the $\Lambda$ processes, a population model that features a skewed
offspring distribution.  We describe the shape of the equilibrium allele
frequency distribution as a function of the model parameters.  We show that the
mutation rates can be uniquely identified from the equilibrium distribution, but
that the form of the offspring distribution itself cannot be uniquely identified.
We also introduce an infinite-sites version of the $\Lambda$ process, and we use
it to study how reproductive skew influences  standing genetic
diversity in a population.  We derive asymptotic formulae for the expected number
of segregating sizes as a function of sample size.  We find that the
Wright-Fisher model minimizes the equilibrium genetic diversity, for a given
mutation rate and variance effective population size, compared to all other
$\Lambda$ processes.
\end{abstract}

\newpage

\section{Introduction}

Many questions in population genetics concern the role of demographic
stochasticity in populations, and its interaction with mutation and selection in
determining the fates of allelic types.  The foundational work of Fisher, Wright,
Haldane, Kimura \cite{Fisher30,Wright31,Haldane,Kimura2} and others has been instrumental in shaping our intuition about the
powerful role that genetic drift plays in evolution, and especially its role in
maintaining diversity.  This classical theory, and the view of genetic drift as a
strong force, emanates from the Wright-Fisher model of replication in a population,
and its large-population limit, the Kimura diffusion \cite{Kimura55}.   The diffusion
approximation has been particularly well-studied, not only because it is mathematically
tractable, but also because it is robust to variation in many of the underlying model
details.  Many discrete population-genetic models, including a large number of
Karlin-Taylor and Cannings processes \cite{pmid16591161,Cannings74,Ewens04}, share the same diffusion limit as the
Wright-Fisher model, and they therefore exhibit qualitatively similar behavior.

Nevertheless, Kimura's classical diffusion is not appropriate in every circumstance.   Its central assumption is the absence of skew in the
reproduction process --- that is, the assumption that no single individual can contribute a sizable proportion to the composition of the population in
a single generation.  Several recent studies have suggested that this assumption may be violated in several species, especially in marine taxa but
also including many types of plants \cite{Beckenbach94,Hedgecock94}, whose mode of reproduction involves a heavy-tailed offspring distribution.

While the number of empirical studies on heavy-tailed offspring distributions is
limited, there is a rich mathematical theory to describe the  dynamics of
populations with heavy reproductive skew.  Beginning with Cannings' 1974 paper on
neutral exchangeable reproduction processes, this literature has led to
generalized notions of genetic drift, which subsume the traditional Wright-Fisherian
concept of drift.  The resulting forward-time continuum limits of such processes
generalize the Kimura diffusion. One tractable class of models are the so-called
$\Lambda$-Fleming-Viot processes, parameterized by a drift measure $\Lambda$.  The
corresponding backward-time, or coalescent theory, for such processes leads to the
$\Lambda$-coalescents, first defined by Pittman and others \cite{Pitman,Sagitov}.   Two
conspicuous features stand out in this more general theory: $\Lambda$ processes
may have discontinuous sample paths, which feature ``jumps'' in the frequency of an
allele, in contrast to the continuous sample paths of Kimura's diffusion.
Likewise, the coalescents of such processes typically exhibit multiple and even
simultaneous mergers, instead of the strictly binary mergers of the classical
Kingman coalescent.  

Although mathematical aspects of such population processes (such as their construction,
existence, uniqueness etc.) have already been described, the specific
population-genetic consequences of reproductive skew have only recently begun to
be worked out.  In many cases, the classical picture of population genetics must
be considerably enlarged to accommodate new phenomena --- see for example \cite{Mohle2} on
generalizations of the Ewens' sampling formula, \cite{EW4} on linkage disequilibrium in
processes with skewed offspring distributions, \cite{Birkner2,BirknerTPB2011} on inference and sampling in the $\Lambda$ coalescent, and \cite{Der2012}
on the fixation probability of an adaptive allele in the $\Lambda$ process.

The purpose of this paper is to study the stationary allele frequency
distribution for populations with reproductive skew, under neutrality.  When there
are a finite number of allelic types subject to mutation, allele frequencies
evolve to a unique stationary distribution, and our principle aim will be to
understand how this distribution depends on the form of reproductive skew, 
$\Lambda$, and how it may depart from the Wright-Fisherian picture.  

Whereas a closed-form expression exists for the stationary allele frequency
distribution in the (continuum) Wright-Fisher model, very few explicit expressions can be
obtained in the general case of an arbitrary $\Lambda$ drift measure.  Instead,
we study the stationary distribution indirectly, first by deriving a recurrence
relation satisfied by the moments in the two-allele scenario.  This relation
provides significant information about how the model parameters $(\theta,\Lambda)$
influence the stationary distribution.  In particular, we demonstrate that the
mutation parameters $\theta$ are identifiable from the stationary distribution,
whereas the form of drift, $\Lambda$, is not in general identifiable.

We also study how reproductive skew alters the standing genetic diversity in a
population at equilibrium.  Some numerical experiments of \citeN{Mohle2}, as well as some
asymptotic results of \citeN{Berestycki1} for the Beta-coalescent, have suggested that
the Wright-Fisher model tends to minimize standing diversity, compared to other offspring
distributions.  To analyze this behavior, we develop a $\Lambda$-version of
Kimura's infinite-sites model, and we study the mean number of segregating sites,
$\mathbb{E}S_n$, in a sample of size $n$.  This measure of genetic diversity is
robust in the sense that it is immune to many assumptions of the model and it
coincides with the mean number of segregating sites in other infinite-sites
models, including Watterson's fully linked infinite-sites model.  We demonstrate
that the Wright-Fisher model minimizes diversity amongst all $\Lambda$-processes
of the same variance-effective population size. In other words, reproductive skew
always tends to amplify standing genetic diversity, compared to the classical
population-genetic model.  We also derive a recursion formula for the mean number
of segregating sites, and we use this to obtain asymptotic formulae for the number
of segregating sites in large samples.

The remainder of the paper is structured as follows.  We start by reviewing
discrete population models under reproductive skew.  We then describe the
forward-time continuum limits of such processes, which can be identified as
$\Lambda$-Fleming-Viot processes.  We develop a recursion equation for the moments
of the stationary distribution in the two-allele case, and we use this to
determine the identifiability of model parameters.  To further examine
equilibrium diversity we then introduce a two-allele, infinite-sites model with
free recombination, and we study the frequency spectrum of samples from this
process.  This leads to a recursion formula for the mean number of segregating
sites, and theorems concerning the minimization and maximization of diversity among
 all $\Lambda$ measures.  We conclude by providing a simple intuition for
our results, and by placing them in the context of the large literature on
reproductive skew.

\section{Discrete population models with reproductive skew}

The $\Lambda$-models are generalizations of the classical Wright-Fisher and Moran
processes, which incorporate the possibility of large family sizes in the
offspring distribution.  The characteristics of these processes are most easily
understood by studying their continuum limits, described below.  Nonetheless, we
shall first describe these models and review their properties in a discrete
setting, along the lines of the treatment in \citeN{EldonWakeley}.

We consider a population containing a fixed number $N$ of individuals, each of two
types.  At every time step, a single individual is chosen uniformly from the
population and produces a random number $U$ offspring, drawn from a distribution
of offspring numbers, $P_U$.  The subsequent generation is then comprised of the
$U$ offspring from the chosen individual supplemented by $N-U$ other individuals,
randomly selected without replacement from the remainder of the population.  Only
a single individual contributes offspring in each reproduction event --- the
remaining individuals who neither contribute offspring nor die simply persist to
the next time step. 

When the offspring distribution $P_U$ is concentrated at two individuals, i.e.
$\mathbb{P}(U = 2) = 1$, this model coincides with the Moran process.  More
generally, we consider any discrete offspring number distribution $P_U$ supported
on the set $\{0,\ldots,N\}$.  

To incorporate mutation, an additional stage is appended after
reproduction wherein each individual may mutate to the opposing type,
independently and identically with a probability that depends upon the
individual's type, $\mu_1/N, \mu_2/N$, with $\mu_i \geq 0$.  This composite
process is graphically depicted in Figure 1.  We shall term this discrete process
a ``generalized'' Eldon-Wakeley model.

\begin{figure}[h]
  \centering
    \includegraphics[width=7in]{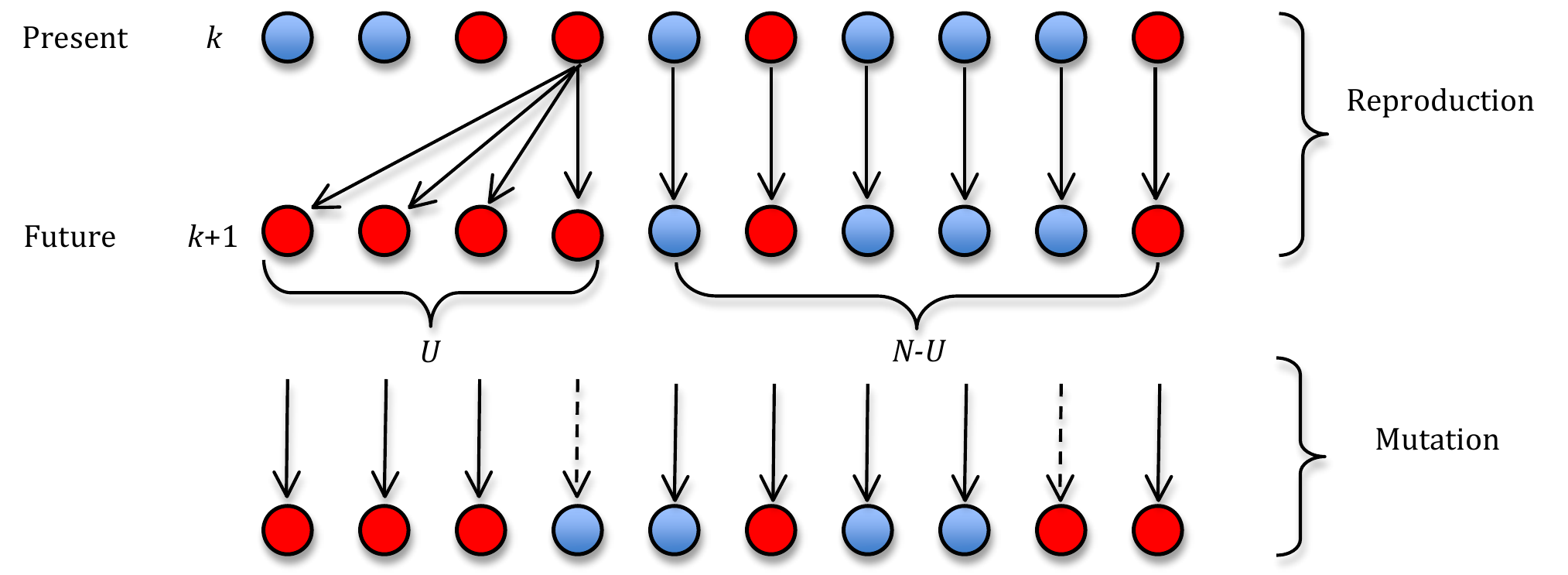}
      \caption{Schematic diagram of a discrete-time population model with reproductive skew.}
\end{figure}

\subsection{The transition matrix.}

We keep track of the number of individuals in generation $k$ of type 1, denoted by
$X_k$.  Since $X_k$ is a Markov chain on the states $\{0,\ldots,N\}$, it
possesses an associated transition matrix 
$\mathbf{P}_{ij}, 0 \leq i,j, \leq N$.  
The transition matrix $\mathbf{P}$ can be written as a product of 
two matrices,
\begin{equation}\label{matdecomp}
\mathbf{P} = \mathbf{Q} \mathbf{M}
\end{equation}
corresponding to the reproduction and mutation stages described above.
The rows of the mutation matrix $\mathbf{M}$ are sums of independent 
binomial distributions, representing the mutational flux from each class. 
When all mutation rates are zero $\mathbf{M} = \mathbf{I}$, the
identity matrix.  The matrix $\mathbf{Q}$ describes the neutral genetic 
drift due to reproduction alone in the $\Lambda$-model; its form is more 
complex, with rows that are mixtures of hypergeometric
distributions whose means depend on the offspring distribution $P_U$. 
An important quantity is the first row of $\mathbf{Q}$, called the ``offspring distribution''
of the process: $\mathbf{Q}_{1,j}$ for $j=0,1, \ldots N$.  The variance of 
this distribution, called the ``offspring variance" $\sigma^2_N$, determines the 
time-scaling of the continuum limit (see below).

\section{Continuum approximations of population models with reproductive skew}

Analysis of the Moran or Wright-Fisher model is facilitated by a continuum
limit, which becomes accurate in the limit of large population size $N \rightarrow
\infty$ \cite{Kimura55,Ewens04}.  As described in \cite{Der2011}, it is  possible to
derive a continuum limit for a significantly larger class of discrete population processes,
including the $\Lambda$ models, without restrictions on
the offspring distribution $P_U$. Similar work in the Cannings case has been developed by \citeN{Mohle01}.
While the limiting continuum processes are not, in general, diffusions with
continuous sample paths, they are still characterized by an operator $G$, the
infinitesimal generator of the continuum process, which reduces to the
second-order differential equation of Kimura in the classical Wright-Fisher case.  

Continuum approximations involve choosing how to scale time and space, as $N\rightarrow \infty$.  Such scalings replace the number $i$ of individuals with the frequency $x = i/b_N$, and the generation number $k$ by the time $t = kc_N$,
for some choices of sequences $\{b_N\}, \{c_N\}$.  The continuum limit is then the process
\begin{equation}
\tilde{X}_t = \lim_{N \rightarrow \infty} \frac{1}{b_N} X_{[t/c_N]}
\end{equation}
In the classical Moran model we use the scalings $b_N = N$ and $c_N = N^{-2}$.
In fact, it can be shown that the relationship between the space-scaling $b_N$ and
time-scaling $c_N$ is fixed, in the sense that no other relationship leads to
non-trivial limiting processes.  We wish study allele frequencies, and
hence impose the natural scaling $b_N = N$.  The general theory \cite{Mohle01,Der2011} then indicates that the time-scaling must be proportional to 
\begin{equation}\label{timescaling}
c_N = \frac{\sigma^2_N}{N}
\end{equation}
where $\sigma^2_N$ is the offspring variance of the $\Lambda$ process.

Once a time-scaling is fixed, then so is the appropriate scaling regime for the
mutation rates, in order to produce a non-trivial balance of mutation and drift. This scaling must satisfy:
\begin{equation} \label{murates} \mu_i = O(\sigma^2_N)
\end{equation}
In the classical Moran model $\sigma^2_N= 1/N$, which produces the traditional scaling of mutation rate $\mu_i = O(N^{-1})$.  In other models, such
as the models of \citeN{EldonWakeley}, where $\sigma^2_N = N^{-\gamma + 1}$ and $\gamma < 2$, mutation rates must scale faster to compensate for the increased rate of evolution from the drift process.

\subsection{The limiting process for a generalized Eldon-Wakeley model.}
By applying the techniques of \cite{Der,Mohle01}, one may derive the continuum limit
for the generalized Eldon-Wakeley process.  These limits are characterized by an operator $G$, and an associated Kolmogorov backward equation, analogous to the diffusion equation of Kimura.   
We assume that we have a sequence of Eldon-Wakeley models, one for each population size $N$, and each with offspring distribution $P_U^{(N)}$.  We assume the time-scaling and mutational constraints of (\ref{timescaling}) and (\ref{murates}) so that
\begin{equation}
 \theta_i = \lim_{N \rightarrow \infty} \frac{2\mu_i}{\sigma_N^2}
 \end{equation}
defines the effective population-wide mutation rate.   Under an appropriate condition on the sequence of offspring distributions $P_U^{(N)}$, there exists a limiting measure $\Lambda$ which may be derived from $\{P_U^{(N)}\}$ as:
\begin{align}
\Lambda  &= \lim_{N \rightarrow \infty} \Lambda_N\\
\Lambda_N(i/N) &= \left(\frac{i}{N}\right)^2 P_U^{(N)}(i), \qquad i = 0,\ldots,N
\end{align} 
and which characterizes the continuum limit.  Letting $\tilde{X}^{(N)}_t = \frac{1}{N} X_{[ t/c_N]}$ denote the time and state re-scaled process, one can show then that $\tilde{X}^{(N)}_t$ converges to a limiting 
process $\tilde{X}_t$ that satisfies the backward equation:

\begin{equation}\label{contrep}
\frac{\partial u(x,t)}{\partial t} = Gu(x,t), \qquad u(x,0)=f(x)
\end{equation}
where
\begin{equation}\label{lambdagen}
G u(x) =    \frac{1}{2}(-\theta_1 x + \theta_2 (1-x)) \frac{\partial u }{\partial x} + \int_0^1 \frac{xu(x+(1-x)\lambda) - u(x) + (1-x) u(x-\lambda x)}{\lambda^2}\, d\Lambda(\lambda)
\end{equation}
and where $u(x,t)=\mathbb{E}[f(\tilde{X}_t)|\tilde{X}_0 = x]$.

The Markov process whose generator $G$ is given by (\ref{lambdagen}) is called the forward-time, two-type $\Lambda$-Fleming-Viot process.

%

\subsection{Intuitive remarks on the generator.}

As with the matrix decomposition of (\ref{matdecomp}), the generator of
(\ref{lambdagen}) splits into two terms: a portion $ \frac{1}{2}(-\theta_1 x +
\theta_2 (1-x)) \frac{\partial }{\partial x}$ that describes mutation, independent
of the reproduction measure $\Lambda$, and an integral portion  that describes
genetic drift.  The term describing mutation coincides with the standard
first-order advection term in Kimura's diffusion equation.  The integral term
however, generally differs from the Kimura term, and it depends on the drift
measure $\Lambda$.

Throughout the remainder of this paper we distinguish several important families
of $\Lambda$-processes. We  define the \emph{pure $\Lambda$ processes}  to be
those models for which $\Lambda = \delta_\lambda$, the Dirac measure concentrated
at a single point $\lambda$, with $0 \leq \lambda \leq 1$.  Since
(\ref{lambdagen}) expresses the generator as an integral decomposition over such
Dirac measures, we can think a $\Lambda$-process as being a random mixture of
these pure processes.  Of particular interest are the extreme cases
$\Lambda = \delta_0$, and $\Lambda = \delta_1$ --- which correspond to the
Wright-Fisher process and the so-called ``star'' processes, respectively.  As we
will show, these two processes constrain the range of dynamics in $\Lambda$
models. Another  well-studied family in the coalescent literature are the
Beta-processes, for which $\Lambda$ has a Beta distribution.

One can interpret $\Lambda$ as a ``jump'' measure controlling the frequency of
large family sizes.  If $\Lambda$ is concentrated near zero, then jump sizes are
small. 
In this regime, the integrand $\frac{xu(x+(1-x)\lambda) - u(x) + (1-x) u(x-\lambda
x)}{\lambda^2}$ behaves like the standard Kimura drift term $\frac{1}{2}x(1-x)
u''(x)$.
For the pure processes, where $\Lambda$ is concentrated at the point
$\lambda$, allele frequencies remain constant for an exponential amount of time,
until a bottleneck event in which a fraction $\lambda$ of the population is
replaced by a single individual.  Such events cause the allele frequency to
increase instantaneously by the amount $(1-x)\lambda$, or decrease by $\lambda x$.  In the most general case of an arbitrary measure $\Lambda$, these behaviors are mixed, and the jump events occur at exponential times with a random mean, and are associated with jumps of random size $\lambda$.  

If $\Lambda$ places large mass near zero, the process becomes diffusion-like, with sample paths exhibiting frequent, small jumps.  On the other hand, if $\Lambda$ is mostly concentrated away from zero then allele dynamics 
are of the ``jump and hold'' type, with fewer, but more sizable, jumps.  Such behavior is most extreme in the star model, whose sample paths are constant until a single jump to absorption.  

\section{The Stationary Distribution of $\Lambda$-Processes}

In the absence of mutation, $\mathbf{M} = \mathbf{I}$, allele frequencies must
eventually fix at 0 or 1, and thus any discrete generalized Eldon-Wakeley model
possesses a trivial stationary distribution whose concentration at the absorbing states
$\{0,N\}$ depends on the initial condition.  When mutation rates $\mu_i$ are strictly positive, however, each
generalized Eldon-Wakeley process in a population size $N$ possesses a unique, non-trivial 
stationary distribution, $\pi_N$, to which the process converges, regardless 
of the initial condition.

In Figure 2, we plot the stationary distributions for a few $\Lambda$-processes,
in the case of symmetric mutation $\theta_1 = \theta_2$.   Generally, these
distributions have the same qualitative dependence on the mutation rate as the
classical Wright-Fisher stationary distribution: they continuously progress
from Dirac singularities at the boundaries to distributions concentrated more in
the center of the interval, as the mutation rate increases.   It is interesting to
observe, however, that the non-Wright-Fisherian processes tend to have more mass
at intermediate allele frequencies, and less relative mass near the boundaries, than the
Wright-Fisherian model.  We shall study this
phenomenon more precisely, below.

\begin{figure}[h] 
\begin{center}
\mbox{
\subfigure[Beta process ($\alpha =0.7, \beta = 1$).]{\includegraphics*[width = 3.2in]
    {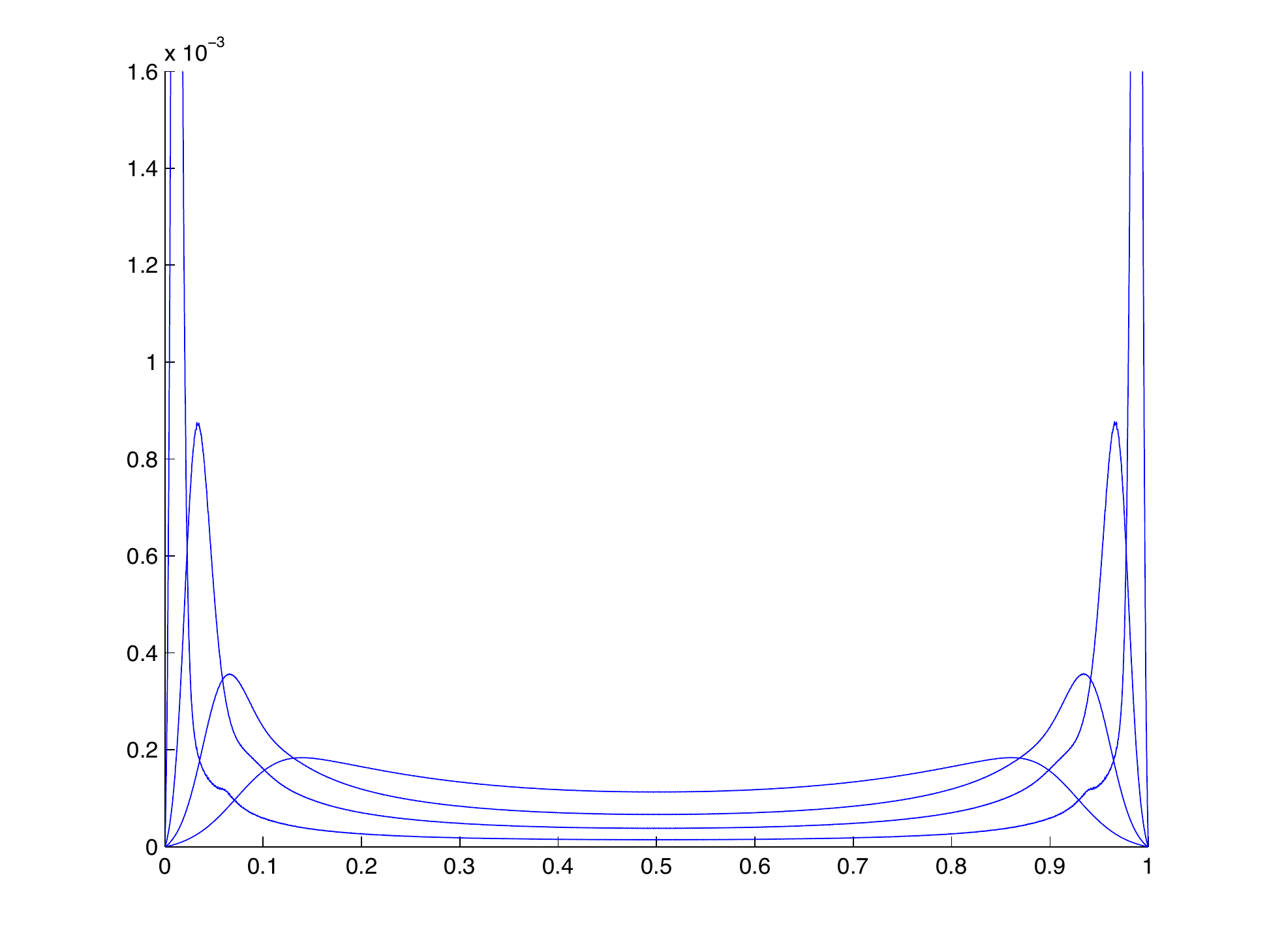}} 
\subfigure[Beta process  ($\alpha = 1.3, \beta = 1$).]{\includegraphics*[width = 3.2in]
    {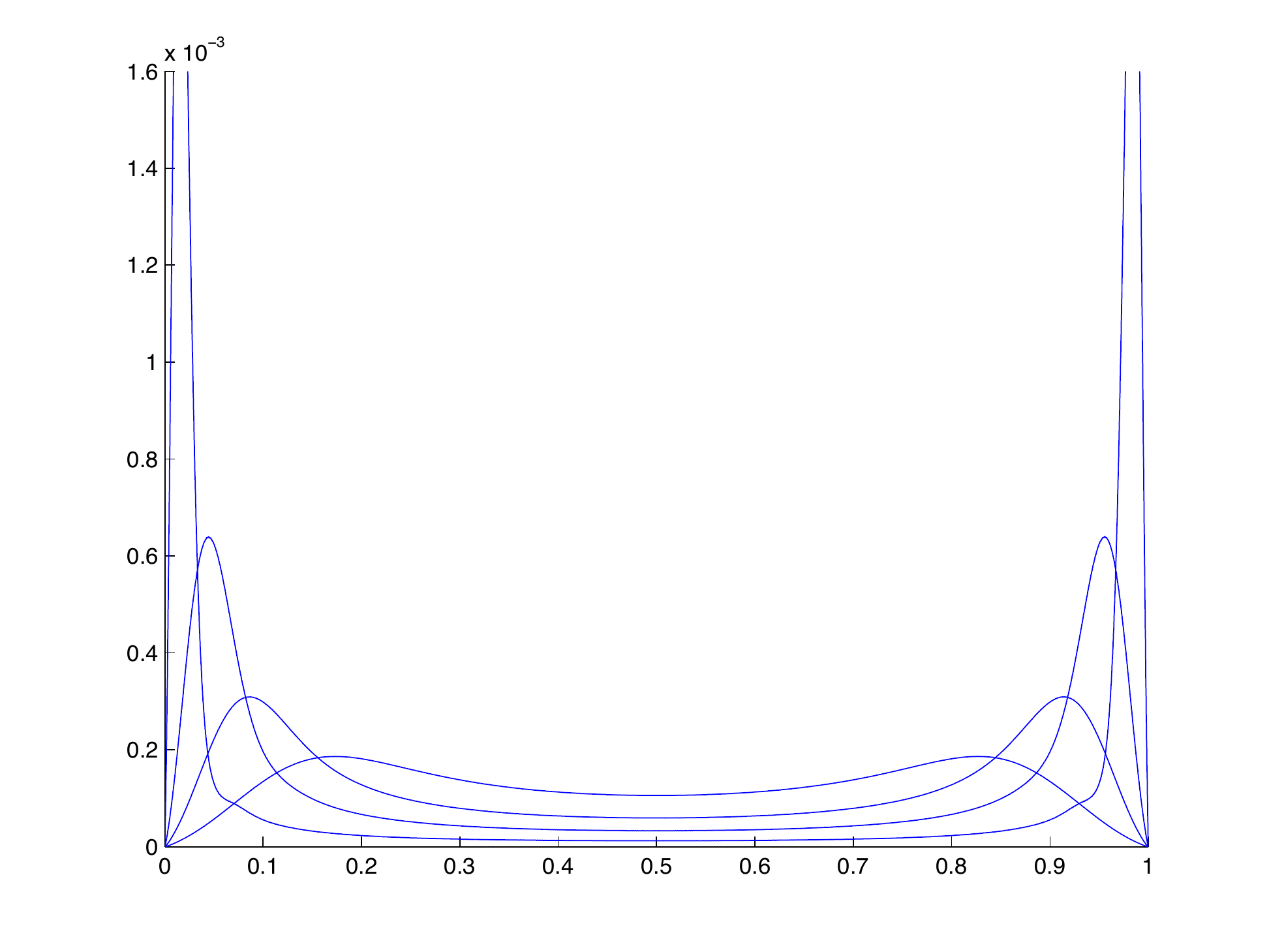}}}
 \subfigure[Wright-Fisher (dashed), Beta (solid) processes, \hspace{1cm} $\alpha = 0.7, \beta = 1$.]{\includegraphics*[width=3.2in]{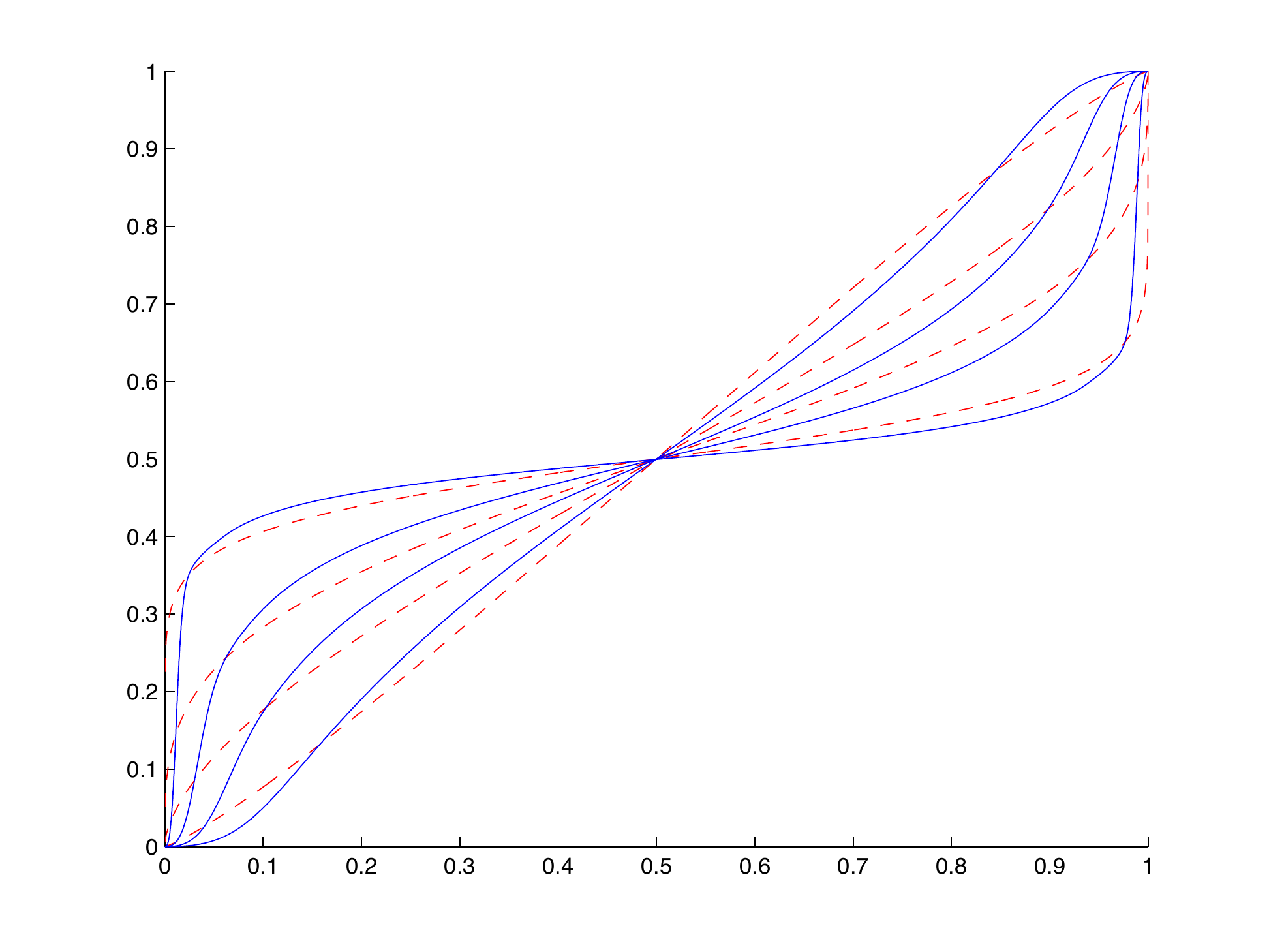}}
 \subfigure[Wright-Fisher (dashed), Beta (solid) processes, \hspace{1cm} $\alpha = 1.3, \beta= 1$.]{\includegraphics*[width=3.2in]{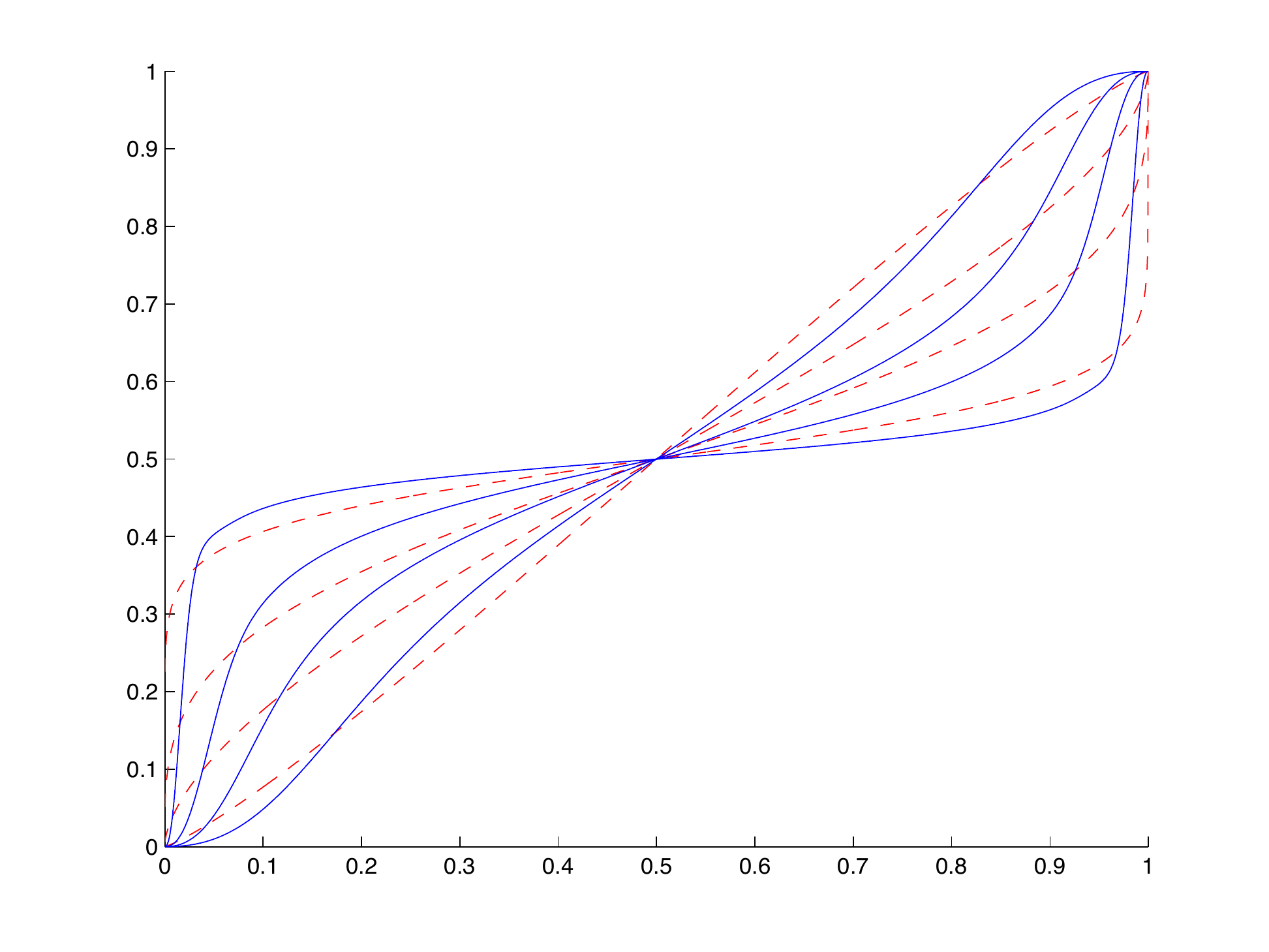}}
\end{center}
\caption{Stationary Distributions for the Wright-Fisher and Beta processes.  Top panels: stationary densities.  Bottom panels, stationary cumulative distribution functions.  Mutation values are $\theta=0.1,0.3,0.6,1.2$.  Population size $N = 8000$.}
\label{samplepaths}
\end{figure}

The continuum limit $\tilde{X}$ of a sequence of Eldon-Wakeley processes also
possesses a unique stationary distribution, $\pi$.  In the Appendix, we
demonstrate that
\begin{equation}
\pi_N \rightarrow \pi, \quad \text{as } N \rightarrow \infty
\end{equation}
In other words, the sequence of discrete equilibrium measures converges to the
continuum equilibrium distribution.  As a result, we can use the continuum equilibrium
as a good approximation in large populations.

\subsection{Moments of the Stationary Distribution.}

In the case of two alleles, the stationary allele frequency distribution 
describes the likelihood of finding the mutant allele at any given frequency, at
some time long in the future.
We will study the moments of the stationary distributions for $\Lambda$ processes
using a version of the Fokker-Planck equation, analogous to the equation used by Kimura to
study the stationary distribution of the Wright-Fisher process.
In general, a stationary distribution
$\pi$ of a Markov process with generator $G$ is the 
solution to its so-called adjoint Fokker-Planck equation, so that
\begin{equation}\label{adeq}
\int_0^1 Gu(x)\, d\pi = 0,
\end{equation}
for every smooth function $u$ on $[0,1]$. We take $G$ as the
generator for the $\Lambda$ process with mutation, given by (\ref{lambdagen}).
Although it is difficult to solve for $\pi$ in general, 
this equation can nonetheless be used to obtain
detailed information about the stationary distribution.

To begin, we develop formulae for the moments of the stationary distribution,
which will allow us to characterize aspects of standing genetic diversity.
Let $m_k$ denote the $k$-th moment of $\pi$, $m_k = \int_0^1 x^k\, d\pi(x)$.  
Setting $u(x) = x$ into (\ref{adeq}) yields an equation for the mean value of the equilibrium, so that 
\begin{equation}\label{moment1}
m_1=\frac{\theta_2}{\theta_1+\theta_2}
\end{equation}
Next, setting $u(x) = x^2$ into (\ref{adeq}) yields a relation between $m_2$ and $m_1$, 
\begin{equation}\label{moment2}
m_2=\frac{(1+\theta_2)\theta_2}{(\theta_1+\theta_2)(1+\theta_1+\theta_2)}
\end{equation}
This recursive process can be continued, because the generator $G$ of
(\ref{lambdagen}) maps polynomials of degree $k$ to polynomials of degree $k$.
Thus, we can derive a system of equations that define the moments of $\pi$.  In
the appendix, we show that this recursion has the form:
\begin{equation} \label{gt}
m_k = \frac{\left(\frac{k}{2} \theta_2 + a_{k-1,k}\right) m_{k-1} + \sum_{j=1}^{k-2} a_{jk} m_j}{\frac{k}{2} (\theta_1 + \theta_2) + a_{kk}}
\end{equation}
where the coefficients $\{a_{jk}\}$, $1 \leq j \leq k$, $k=1,2,\ldots$, are functions of  $\Lambda$, and are given by
\begin{align}\label{aj1}
a_{kk} &= \int_0^1 \frac{1-(1-\lambda)^k - k\lambda(1-\lambda)^{k-1} }{\lambda^2} \, d\Lambda(\lambda)\\
a_{j,k} &= \binom{k}{j-1} \int_0^1 \lambda^{k-j-1} (1-\lambda)^{j-1}\, d\Lambda(\lambda), \qquad j =1,\ldots,k-1\label{aj2}
\end{align}
Initializing this system by (\ref{moment1}), and observing that $a_{kk} > 0$, 
we see that (\ref{gt}) uniquely determines the moments of $\pi$, and indeed this
equation can be used to solve for any specific moment of the stationary
distribution.   While it does not appear that the moments $m_k$ can be solved explicitly 
to produce simple, closed-form expressions as functions of the $\Lambda$ measure, 
it is clear that the coefficients $a_{jk}$ are all linear combinations of moments of $\Lambda$.   
Moreover, each moment $m_k$ is always a ratio of polynomials in $\theta_1$ and $\theta_2$.



\subsection{Identifiability of parameters from equilibrium.}

One of the most important questions about the stationary allele frequency distribution is what
population-genetic parameters can be identified from it --- that is, which
parameters of the population can be uniquely determined from data sampled in
equilibrium?
In the case of $\Lambda$ processes, the parameters we
might wish to infer are the mutation rates, $\theta_1$ and $\theta_2$, as well as 
the (high-dimensional) drift
measure, $\Lambda$, which describes the offspring distribution.


The first two moments of the stationary distribution are given by (\ref{moment1}) and
(\ref{moment2}), and they are independent of the drift measure, $\Lambda$.  Thus,
the first- and second-order moments of the stationary distribution for all $\Lambda$-processes, and any
function of these moments, such as the second-order heterozygosity $\int_0^1
x(1-x) d\pi(x)$, must coincide with those of the classic Wright-Fisher model.   In
the case of symmetric mutation $\theta_1 = \theta_2 = \theta$, the \emph{third}
moment is also, remarkably, constant across the $\Lambda$-processes, and has the
value
\begin{equation}
m_3 = \frac{2+\theta}{4+8\theta}.
\end{equation}

The constancy of the first two moments with respect to $\Lambda$, and the fact
that the mapping from the first two moments to the two mutation parameters $(m_1,m_2) \mapsto (\theta_1,\theta_2)$ 
given by (\ref{moment1}) and (\ref{moment2}) is one-to-one, allows us
to conclude that, regardless of the underlying reproductive process, the mutation
rates $(\theta_1,\theta_2)$ are always identifiable from the equilibrium
distribution.  This is a tremendously productive result --- because it means that we can
always infer mutation rates from sampled data, even when the offspring
distribution of a species is unknown to us.

Conversely, we may ask whether we can identify the form of reproductive process
without knowledge of mutation rates --- that is, is $\Lambda$ uniquely identifiable
from the stationary distribution alone?  It turns out that the answer is negative,
as we demonstrate with the following simple example. 

Consider the  star process $\Lambda = \delta_1$ with mutation parameters $\theta_1= \theta_2=\theta$.  
The generator for this process is
 \begin{equation}\label{lcsta}
Gu(x) = \frac{1}{2}\theta(1-2x)u'(x) + (1-x)u(0)-u(x) + xu(1)
 \end{equation}
The associated stationary distribution $\pi_1$ is easily derived (see Appendix),
and it has a density $d\pi_1/dx$ given by 
\begin{equation}\label{statexcan}
\frac{d\pi_1}{dx} = \frac{1}{\theta}|1-2x|^{\frac{1-\theta}{\theta}}
\end{equation}

For comparison, the Kimura diffusion has a Dirichlet-type stationary distribution
$\pi_0$: \begin{equation}\label{WFequilibrium} \frac{d\pi_0}{dx} =
\frac{\Gamma(2\theta)}{(\Gamma(\theta))^2} x^{\theta - 1} (1-x)^{\theta - 1}
\end{equation}
 
Note that both distributions (\ref{statexcan}) and (\ref{WFequilibrium}) coincide
when $\theta=1$, despite the enormous difference between the drift measures of
these processes. Hence, the map from a given $\Lambda$ process to its stationary
distribution is not one-to-one, and consequently the drift measure $\Lambda$
cannot generally be identified from the stationary distribution. This is a
pessimistic result, because it implies that that the offspring distribution and
form of genetic drift cannot in general be inferred from data collected in
equilibrium, even when the mutation rate is known.

\section{An infinite-sites model for the $\Lambda$-processes}

In order to understand how reproductive skew influences standing genetic
diversity, we now develop an infinite-sites version of the $\Lambda$ process and
study its equilibrium behavior.  This model generalizes the infinite-sites
approach of \cite{Arindam,Desai+Plotkin08}, for the Wright-Fisher model.  We will
study the sampled site frequency spectrum of our model, under two-way mutation.
Our analysis will allow us to quantify our previous observation that the
Wright-Fisher model minimizes the amount of standing genetic diversity, amongst all
$\Lambda$ processes. The site frequency spectrum that we will describe in this
section, for independent sites, differs from the Watterson-type spectrum for fully
linked sites; but our approach nonetheless yields information in that case as
well.

We consider an evolving population of large size $N$, following the reproduction
dynamics of a neutral forward-time $\Lambda$-process, for a fixed $\Lambda$
measure.  We keep track of $L$ sites along the genome, each with two possible
allelic types under symmetric two-way mutation at rates $\theta = \theta_1 =
\theta_2$.  The allele dynamics at each site are described by a two-type
$\Lambda$-process; and the site processes are assumed independent of one another
(that is, we assume free recombination).

Let $\pi_{\theta}$ denote the two-allele stationary distribution for the $\Lambda$ model, given by (\ref{adeq}),  
where the subscript denotes the explicit dependence on the mutation rate.  We imagine sampling $n$ individuals from 
the population at equilibrium, assuming $n \ll N$.  We let  $Y_i$, $1 \leq i \leq L$ represent
the (random) number of sampled individuals at site $i$ with a particular allelic type, so that
their joint distribution has the form
\begin{equation}
P(Y_1 = y_1,\ldots,Y_L = y_L) = \prod_{i=1}^L \int_0^1 \binom{n}{y_i}x^{y_i}(1-x)^{n-y_i}\,d\pi_\theta(x)
\end{equation}

The sampled site frequency spectrum \cite{Sawyer+Hartl92,Bustamanteetal01} is
defined as the vector $(Z_0,\ldots, Z_n)$ \begin{equation}\label{segsite} Z_k =
\sum_{i=1}^L 1_{Y_i = k}, \quad k=0,\ldots,n \end{equation} The variables $Z_k$
record the number of sites with precisely $k$ (out of $n$) sampled individuals of
a given allelic type. In this sense, the sampled site frequency spectrum
represents a discretized version of the stationary distribution $\pi_\theta$.  The
variables $(Z_0,\ldots,Z_n)$ are distributed multinomially on the simplex
$\sum_{k=0}^n Z_k = L$.  The sites $Z_1,\ldots,Z_{n-1}$ are called the
\emph{segregating sites}, representing locations where there is diversity observed
in the the sample.  Conversely, the sum $Z_0+Z_n$ represents the number of
monomorphic sites in the sample.

\subsection{The infinite-site limit and its Poisson representation.}

To study the sampled site frequency spectrum we take the limit of an infinite number of sites, $L\rightarrow \infty$, and we apply a 
Poisson approximation.  We define the genome-wide mutation rate as $\Theta_L = L
\cdot \theta$, and we assume that this mutation rate approaches 
a constant in the limit of many sites: $\Theta_L \rightarrow \Theta < \infty$.  In the Appendix, we show that the segregating site 
variables $(Z_1,\ldots,Z_{n-1})$ then converge, as $L\rightarrow \infty$, to a sequence of independent Poisson random variables with 
means $(c_1\Theta,\ldots,c_{n-1}\Theta)$, given by
\begin{equation}\label{cjs}
c_j(n) = \lim_{\theta \rightarrow 0} \frac{1}{\theta} \int_0^1 \binom{n}{j} x^{j} (1-x)^{n-j}\,d\pi_\theta(x), \qquad j=1,\ldots,n-1
\end{equation}

The numbers $c_j$ may be interpreted as an infinite-sites sample frequency spectrum.  From (\ref{cjs}), it is apparent 
that the means $c_j$ 
depend on the heterozygotic moments of $\pi$, and thus, also on the moments of $\Lambda$.   
This representation is thus a generalization of 
a result of \citeN{Arindam} for the two-allele Wright-Fisher independent-sites model, where  $\Lambda = \delta_0$, and
where the spectrum $c_j$ has the form $c_j=\frac{1}{2}\frac{n}{j(n-j)}$ for $j=1,\ldots,n-1$.  

For the Wright-Fisher model, \citeN{Arindam} has shown that the number of segregating sites in the sample of size $n$, $S_n = \sum_{i=1}^{n-1} Z_i$, is a sufficient 
statistic for $\Theta$, under the independent-sites assumption. 
This is an important result because the number of segregating sites vastly
compresses the information in the frequency spectrum,
yet nonetheless contains no loss of information for the purposes of inferring the mutation rate. 
The Poisson representation of the sample frequency spectrum we have derived shows that $S_n$ remains Poisson distributed even in the general $\Lambda$ infinite-sites case --- under the assumption of site independence. 
Thus, the sufficiency of $S_n$ for $\Theta$ remains 
true, and consequently $S_n$ possesses
desirable qualities for robust estimation of $\theta$.  

\subsection{Diversity amplification and the number of segregating sites.}

The number of segregating sites in a sample is a classic and powerful method to 
quantify genetic diversity in a population. Here we study how $S_n$ depends on the
form of reproduction --- that is, on the form of the drift measures $\Lambda$. In
particular, we will show that the Wright-Fisher model minimizes the expected number
of segregating sites in a sample, compared to all other $\Lambda$ processes.
Thus, large family sizes in the offspring distribution will tend to amplify the
amount of diversity in a population.


Under the infinite-sites Poisson approximation, the number of segregating sites
$S_n$ in a sample of size $n$ is Poisson-distributed, and its expected value is
\begin{equation}
\mathbb{E}S_n = \mathbb{E} \sum_{j=1}^{n-1} Z_j = \Theta \sum_{j=1}^{n-1} c_j(n)
\end{equation}
where $c_j(n)$ are the coefficients in (\ref{cjs}).  The binomial theorem applied to (\ref{cjs}) then shows that:
\begin{equation}\label{ww}
\mathbb{E}S_n = \lim_{\theta \downarrow 0} \frac{1}{\theta} \int_0^1 [1-x^n - (1-x)^n] \, d\pi_\theta(x)
\end{equation}
and so we may interpret the expected number of segregating sites as a type of
higher-order heterozygosity statistic of the
stationary distribution. According to  (\ref{cjs}), $c_j(n)$ is a linear combination of moments of $\pi$, of order at most $n$. 
It follows that the average number of segregating sites may be evaluated by the recursion (\ref{gt}) 
and it can be expressed as rational functions of moments of $\Lambda$. The first
several such expressions are listed below:
\begin{align}
\mathbb{E}S_2 &= \Theta\\
\mathbb{E}S_3 &= \frac{3}{2} \cdot \Theta\\
\mathbb{E}S_4 &= \frac{\int_0^1 (5\lambda^2 - 14 \lambda + 11)\, d\Lambda(\lambda)}{\int_0^1 (6-8 \lambda + 3 \lambda^2) \, d\Lambda(\lambda)} \cdot  \Theta\\
\mathbb{E}S_5 &= \frac{5}{2}\frac{\int_0^1 (5-6\lambda+2\lambda^2)\, d\Lambda(\lambda)}{\int_0^1(6-8\lambda+3\lambda^2)\, d\Lambda(\lambda)} \cdot \Theta\\
\mathbb{E}S_6 &= \frac{1}{2} \frac{\int_0^1 (2608\lambda^2-1558\lambda+411-2428\lambda^3+1312\lambda^4-388\lambda^5+49\lambda^6)\, d\Lambda(\lambda)}{\int_0^1 (6-8\lambda+3\lambda^2)(15-40\lambda+45\lambda^2-24\lambda^3+5\lambda^4)\, d\Lambda(\lambda)} \cdot                                          \Theta\\
\vdots \nonumber
\end{align}

These expressions for the expected number of segregating sites become
extremely complex for larger sample sizes $n$. Nevertheless, we can use asymptotic
methods to study how diversity is expected to behave in large sample sizes.  We
will address two primary questions. First, how does does the expected number of
segregating sites, $\mathbb{E}S_n$, grow as a
function of the sample size, for a given drift-measure $\Lambda$?  And second,
which reproduction processes $\Lambda$ maximize and minimize $\mathbb{E}S_n$, for
fixed $\Theta$?  

In the appendix, we use the moment recursion (\ref{gt}) to derive the following
recursion for the sequence  $\{\mathbb{E}S_n\},n=2,3,\ldots$
\begin{equation}\label{srecur}
\mathbb{E}{S}_n = \frac{n \Theta}{2a_{nn}} + \sum_{j=1}^{n-1} \frac{a_{jn}}{a_{nn}} \mathbb{E}S_j
\end{equation}
with the initial value $\mathbb{E}S_1 = 0$, and where $\{a_{jn}\}$ are given by (\ref{aj1}) and (\ref{aj2}).
We can use this relation to obtain detailed information about $\mathbb{E}S_n$
both as a function of the sample size $n$, and as a function of the underlying $\Lambda$ measure.  

Consider first the pure $\Lambda$ processes $\Lambda = \delta_\lambda$, in
which a single individual may replace a given fixed fraction $0<\lambda\leq 1$ of the population.  
Then we can prove from (\ref{srecur}) that (see Appendix):
\begin{equation}\label{sk}
\mathbb{E}S_n = \frac{\lambda \Theta}{2} n + O(\log n), \qquad n \rightarrow \infty
\end{equation}
Two features are of interest in this asymptotic expression.  First, the average number of segregating sites grows linearly 
with sample size in a pure $\Lambda$ process, as opposed to
logarithmically as in the Wright-Fisher case.  Second, 
the rate of linear growth depends on the jump fraction, $\lambda$, so that asymptotically, diversity is maximized for 
large replacement fractions $\lambda$, and correspondingly minimized when this fraction is small.

Equation (\ref{sk}) can be generalized to a larger class of $\Lambda$ measures.
If $\Lambda$ is any probability measure whose support excludes 
a neighborhood of zero, then we have the asymptotic formula:
\begin{equation}\label{asympform}
\mathbb{E}S_n = C(\Lambda) \Theta \cdot n + O(\log n), \qquad n \rightarrow \infty
\end{equation}
where $C(\Lambda)=(2\int_0^1 \lambda^{-1}\, d\Lambda)^{-1}$.  
This equation shows that linear growth of the expected number of segregating
sites is characteristic of any $\Lambda$ process whose drift
measure is bounded away from zero --- that is, any $\Lambda$ process that does not
contain a component of the Wright-Fisher process. This result allows us to 
determine which reproduction processes  $\Lambda$ maximize and minimize the
average number of segregating sites $\mathbb{E}S_n$, for a given value of
$\Theta$.  In the appendix, we prove the following optimization principle:
\emph{for each sample size $n$,  the diversity maximizing and minimizing processes
within the class of all $\Lambda$ processes must in fact be pure $\Lambda$ 
processes, i.e. where $\Lambda$ is concentrated at a single point.}  It follows
then from (\ref{sk}) that, asymptotically, the Wright-Fisher model ($\lambda = 0$)
minimizes, and the star-model ($\lambda = 1$) maximizes, respectively, the mean
number of segregating sites amongst all $\Lambda$-processes.


Although these results apply in the limit of large sample sizes, we conjecture
that the Wright-Fisher and star models are also the extremal diversity processes
for any sample size, $n$.  From the optimization principle stated above, it
suffices to check this statement within the restricted class of \emph{pure}
$\Lambda$ processes. In Figure 3, we show $\mathbb{E}S_n/\Theta$ as a function
of the jump-size parameter $\lambda$ for the pure models, for a few
values of $n$.  These results confirm that the Wright-Fisher model minimizes
$\mathbb{E}S_n$, whereas the star model maximizes $\mathbb{E}S_n$,  over all
$\Lambda$-models.  We have conducted numerical studies which support this
proposition more generally, even for very small sample sizes.   In this sense, the
Wright-Fisher model and star models are extremal processes, and, for a given
effective variance population size, respectively minimize and maximize the
expected genetic diversity in any sample.  
\begin{figure}[h] \centering \includegraphics[width=4in]{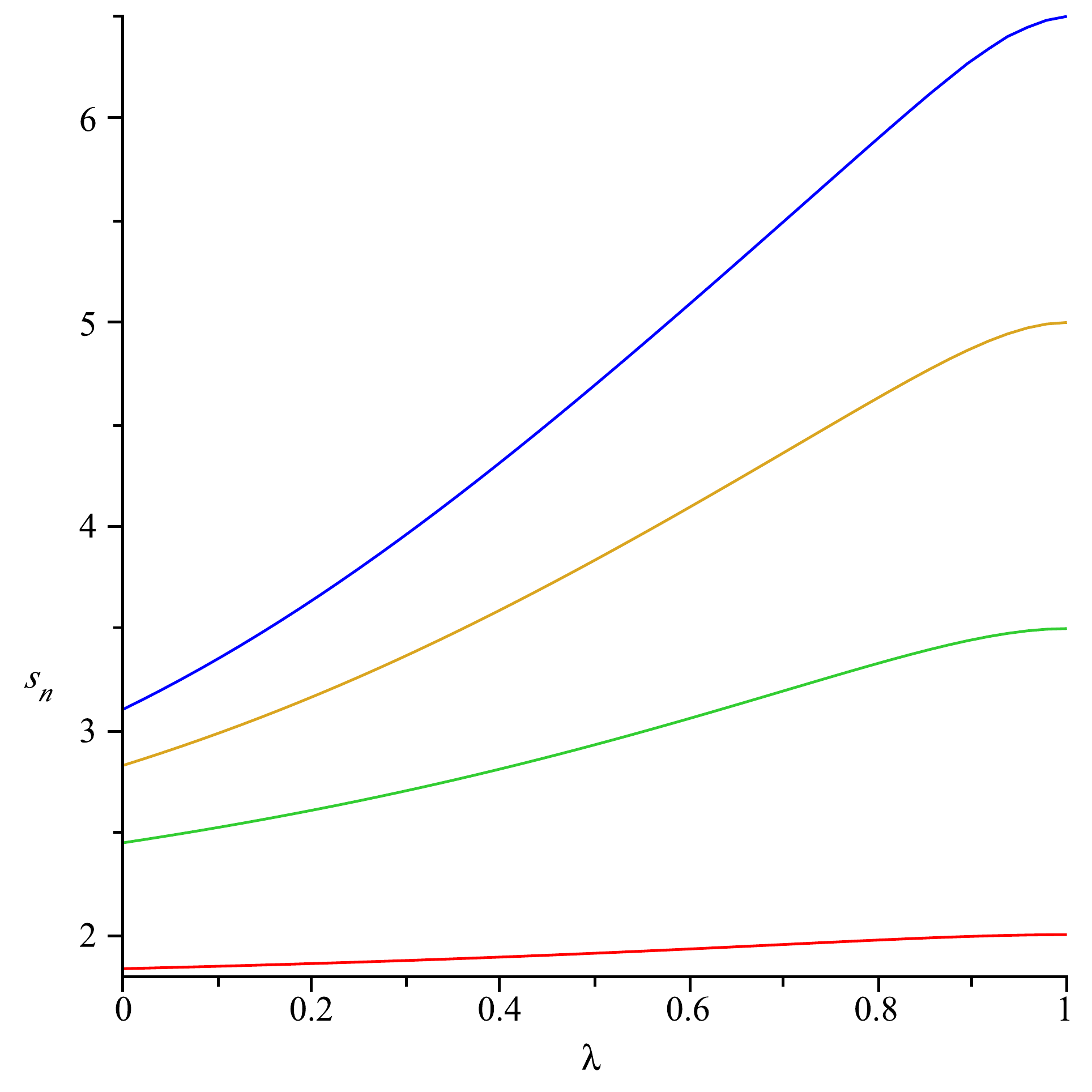}
\caption{Diversity  $s_n=\mathbb{E}S_n/\Theta$ versus jump-size $\lambda$ in the
pure $\Lambda$ processes, for $n=4,7,10,13$ (bottom to top).  These functions
attain their extrema at the endpoints $0$ and $1$, implying that the Wright-Fisher
and star processes minimize and maximize diversity for these sample sizes.}
\end{figure}


\section{Discussion}

We have studied the stationary distribution of a very general class of population
models, under mutation.
We have focused on understanding the interaction between the form of the offspring
distribution and resulting form of genetic drift it engenders, as well as the shape of
the stationary distribution. We have demonstrated that the mutation rate can
always be uniquely identified from the stationary distribution, even when the 
drift measure is unknown (as it always will be, in practice). However, the
form of the drift measure $\Lambda$ cannot always be uniquely identified from equilibrium
properties of the process --- and so may require dynamic data to determine its specific form.

The stationary allele frequency distribution of the Wright-Fisher process is
extremal, in a sense, within the class of $\Lambda$-processes. Specifically, the
Wright-Fisher model exhibits greater probability mass near very high and low
allele frequencies. This observation was formalized by analyzing a $\Lambda$
infinite-sites model, in which we found that the mean number of segregating sites
in a sample is indeed minimized by the Wright-Fisher process.

Our results can be placed in the context of a nascent literature that
views the Wright-Fisher process as an extremal model within the large space of 
possible population processes.  Aside from the diversity-minimization property we
have demonstrated here, it has previously been observed, for instance, that among
$\Lambda$-processes, the Wright-Fisher model minimizes the fixation probability of
an adaptive allele \cite{Der2012}, 
minimizes the time to absorption for new mutants
\cite{Der2011}, and, among generalized coalescents, possesses the fastest rate of
``coming down from infinity'' \cite{Berestycki2010}.  The basic intuition behind all these results 
revolves around the type of sample paths possessed by different processes.
 A typical sample path in the Kimura diffusion
undergoes  a  high frequency of small jumps (in fact, is continuous), and thus new
mutants persist for only
 $O(\log N)$ generations before being eliminated by genetic drift.  By contrast,
in a general $\Lambda$ model with the same variance effective population size,
large jumps in the sample path may occur, but with lower frequency, thereby
lengthening the absorption time --- for example, up to order $N$ generations  in
the pure $\Lambda$ processes.  Since the mean number of segregating sites in
the entire population is the product of the genomic mutation rate and the
expected absorption time for a new mutant, standing genetic diversity must increase
when reproductive skew is present.

Although we have presented results only within the class of $\Lambda$-processes,
many of our formulae --- for example  (\ref{gt}) -- can be generalized to the set
of all Cannings models.  We expect the diversity-minimization property of the
Wright-Fisher model will hold even within this larger family.  


The infinite-sites model of the $\Lambda$-process we have developed here differs from
the Watterson infinite-sites model typically encountered in coalescent theory,
in two respects. First, we have assumed free recombination and hence independent sites,
whereas in Watterson's model sites are tightly linked.  Second, we assume two-way
mutation between alternative alleles at each site, whereas Watterson's model
features one-way mutation at each site away from the existing type.
Nonetheless, some of the results derived for our site-independent, infinite-sites model extend to
the Watterson, linked infinite-sites $\Lambda$ processes as well.  

In general, the (random) number
of segregating sites $S_n$ in a sample is a function of the dependency structure among
sites.  For example, in the simple Wright-Fisherian case, independence of sites
gives rise to a Poisson distribution for $S_n$, compared to a sum of geometric random
variables in the case of no recombination \cite{Ewens04}.  However, as \citeN{Watterson75}
has already remarked, the \emph{mean} value of $S_n$ is generally robust to the
recombination structure of an infinite-sites model.  If $Y_1, \ldots, Y_L$ denote
the allelic distributions at $L$ sites, then (\ref{segsite}) shows that the
expected number of segregating sites is a function only of the marginal
distributions of $Y_i$, instead of their joint distribution.  Thus the expected
number of segregating sites in a sample is unaffected by linkage. Likewise, the
distinction between one-way and two-way mutation (and folded and unfolded spectra)
does not alter the mean number of segregating sites other than by a possible overall scaling.


Because $S_n$ is such a common measure of genetic diversity, our results have some
connections to the literature on
$\Lambda$-coalescents.  Recently, \citeN{Berestycki12arxiv}  showed that, for
those $\Lambda$ measures whose coalescent comes down from infinity, the
(random) number of segregating sites $S_n$ in a sample of size $n$ for the
Watterson model has the asymptotic law:
\begin{equation}\label{berresult}
\frac{S_n}{\int_0^n q\psi^{-1} (q) dq} \rightarrow \Theta
\end{equation}
where $\psi$ is the Laplace exponent of the $\Lambda$ measure, defined as:
\begin{equation}
\psi(q) = \int_0^1 \frac{\exp(-q\lambda) - 1 + q\lambda }{ \lambda^2} d\Lambda
\end{equation}
The authors conjectured that (\ref{berresult}) holds more generally, even
when $\Lambda$ does not come down from infinity.  In this respect, our asymptotic
result (\ref{asympform}) for $\mathbb{E}S_n$ --- derived for $\Lambda$ measures
bounded away from zero (and thus always fail to come down from infinity) is
evidence in favor of their more general conjecture, in the case not covered by the hypotheses of their
theorem.  For under such assumptions, the Laplace exponent has the expression 
\begin{equation}
\psi(q) \sim q \int_0^1 \lambda^{-1}\, d\Lambda
\end{equation}
which implies from (\ref{berresult}) that
\begin{equation}
\int_0^n q \psi^{-1} (q) \, dq \sim n \cdot  \left(\int_0^1 \lambda^{-1} d\Lambda \right)^{-1}
\end{equation}
which is proportional with our formula (\ref{asympform}) for $\mathbb{E}S_n$.  Finally, returning to the case 
of independent sites, developed in this paper, it is also true that the distributional 
convergence of (\ref{berresult})
holds, a fact which follows from the Poisson representation for $S_n$.  


In our analysis of the expected number of segregating sites, we have concentrated
on the two extreme cases--- the Wright-Fisher case, for which $\mathbb{E}S_n$ is
known to grow logarithmically in the sample size $n$, and the case of
pure $\Lambda$ processes (and more generally those $\Lambda$ processes whose
drift measure support excludes zero), for which we have demonstrated linear growth of
$\mathbb{E}S_n$.   Nevertheless, the recurrence relation
(\ref{srecur}) can be used to analyze intermediary cases as well, for example the Beta
processes, in which the density of $\Lambda$ behaves like a power-law in the
vicinity of zero.  For such reproduction measures, growth in
diversity with sample size will lie somewhere between the logarithmic and the
linear cases.

\section{Appendix }

%

\subsection{The Stationary Distribution of Processes with Reproductive Skew.}
Let $X^{(N)}$ be a sequence of discrete generalized Eldon-Wakeley processes, one for each population size $N$, converging to a continuum $\Lambda$ process $\tilde{X}$, under the state and time re-normalization of (\ref{timescaling}), (\ref{murates}).  If we suppose that each discrete process operates under strictly positive mutation rates $\mu_i^{(N)}$, then it is easily verified that the associated forward-time transition matrices $\mathbf{P}^{(N)}$ possess strictly positive entries, and thus, from the Perron-Frobenius theorem, there exists a unique stationary distribution $\pi_N$ for each process $X^{(N)}$.  A standard argument, using the fact that the sequence $\pi_N$ is tight, shows that there is a subsequence $\pi_{N_k}$ converging to a probability measure $\pi$ which is a stationary distribution for $\tilde{X}$ (see \citeN{Ethier+Kurtz86}, for example).   This argument indeed demonstrates that any weak limit point of $\pi_N$ is a stationary distribution $\tilde{X}$; below, through the moment recursion, this distribution is uniquely characterized, and hence every weakly convergent subsequence of $\pi_N$ converges to $\pi$, thus $\pi_N \rightarrow \pi$.  

\subsection{Derivation of a Recursion for the Moments of the Stationary Distribution.\\}

Let $\pi$ be the stationary distribution for the two-type $\Lambda$ process, which satisfies (\ref{adeq}).  In this section we obtain a recursion formula for the moments of $\pi$.  

Define the operator $Lu(x) = xu(x+(1-x)\lambda) - u(x) + (1-x) u(x-\lambda x)$.  
Setting $u(x) = x^k$, $k \geq 0$ we have
\begin{align}\nonumber
Lx^k &= ((1-\lambda)^k - 1)x^k-(1-\lambda)^k x^{k+1} + x(x+(1-x)\lambda)^k\\ \nonumber
&= ((1-\lambda)^k - 1)x^k - (1-\lambda)^k x^{k+1} + x^{k+1}(1-\lambda)^k + kx^k (1-\lambda)^{k-1}\lambda + \sum_{j=0}^{k-2} \binom{k}{j} x^{j+1} (1-\lambda)^{j} \lambda^{k-j} \\
&=((1-\lambda)^k -1+k\lambda(1-\lambda)^{k-1}) x^k + \sum_{j=1}^{k-1} \binom{k}{j-1} x^j (1-\lambda)^{j-1} \lambda ^{k-j+1}\label{expand1}
\end{align}
Separating the $\Lambda$ generator (\ref{lambdagen}) into the mutation and pure-drift portions, we define the latter to be the operator
\begin{equation}\label{pured}
G_Du(x) =\int_0^1 \frac{1}{\lambda^2} Lu(x) d\Lambda(\lambda).
\end{equation}
If we write 
\begin{equation}\label{expand2}
G_D x^k = \sum_{j=1}^k b_{jk} x^j
\end{equation}
Then substituting (\ref{expand1}) into (\ref{pured}), and then comparing the coefficients to (\ref{expand2}) , we have
\begin{align}
b_{kk} &= \int_0^1 \frac{(1-\lambda)^k + k\lambda(1-\lambda)^{k-1} - 1}{\lambda^2} \, d\Lambda(\lambda)\\
b_{jk} &= \binom{k}{j-1} \int_0^1 \lambda^{k-j-1} (1-\lambda)^{j-1}\, d\Lambda(\lambda), \qquad j < k
\end{align}
Let $\pi$ be any stationary distribution of the process.  Then according to (\ref{adeq}), 
\begin{equation} \int_0^1 \left[\frac{1}{2} (-\theta_1 x + \theta_2(1-x)) k x^{k-1} + G_Dx^k \right] d\pi(x) = 0
\end{equation}
Using the expansion for $G_D x^k$ above, we derive:
\begin{equation} \label{momrecur}
\left(b_{kk} - \frac{k}{2} (\theta_1 + \theta_2)\right) m_k + \left(\frac{k}{2} \theta_2 + b_{k-1,k}\right)m_{k-1} + \sum_{j=1}^{k-2} b_{jk} m_j = 0
\end{equation}
where $m_j$ is the $j$-th moment of the $\pi$.  This is equivalent to formulae (\ref{gt}),  (\ref{aj1}) and (\ref{aj2}), where $b_{kk} = -a_{kk}$, and $b_{jk} = a_{jk}$ for $j < k$.

\subsection{Derivation of the star-process stationary distribution.}

Consider the probability measure $\mu$ on $[0,1]$, with density\begin{equation}
\frac{d\mu}{dx} = \frac{1}{\theta}|1-2x|^{\frac{1-\theta}{\theta}}
\end{equation}

The generator for the star-process undergoing symmetric mutation is $Gu = \frac{1}{2}\theta(1-2x)u'(x) + (1-x)u(0)-u(x) + xu(1)$, and the space of twice continuously differentiable functions $C^2[0,1]$ is a core for $G$.  Noting that the density $d\mu/dx$ satisfies the equation $-d/dx(\frac{\theta}{2} (1-2x) d\mu/dx) - d\mu/dx = 0$ and integrating by parts, it is readily verified that $\int_0^1 Gu \,d\mu = 0$ for every $u \in C^2$.  Thus $\mu$ is a stationary distribution for the process, and is further unique as established by the moment recursion (\ref{momrecur}).  

\subsection{Poisson representation of the infinite-sites model.}

In this section we show that the segregating site variables $(Z_1,\ldots,Z_{n-1})$ in the independent sites $\Lambda$ model converge to a sequence of Poisson random variables with means $(c_1\Theta, \ldots, c_{n-1} \Theta)$, given by (\ref{cjs}).   We make use of the structure of the moments of the stationary distribution as found in (\ref{gt}).  First we require a preliminary lemma.
\begin{lemma}\label{baselemma} Let $\pi_\theta$ be the stationary distribution of a $\Lambda$ process undergoing symmetric mutation $\theta$.  Then there exist constants $c_j \geq 0$, for $1 \leq j \leq n-1$,
\begin{equation}\label{gg3}
p_j(\theta) =\int_0^1 \binom{n}{j}x^j(1-x)^{n-j}\,d\pi_{\theta}(x) =  c_j \theta + o(\theta), \quad \theta \downarrow 0
\end{equation}
\end{lemma}
\begin{proof} The lemma is equivalent to saying $p_j(0)=0$ and $p_j$ has a derivative at $\theta = 0$.
First observe from (\ref{gt}) that under $\theta_1 = \theta_2$, all the moments $m_k(\theta)$ of the stationary distribution are differentiable in $\theta$ for all $\theta \geq 0$, and hence $p_j(\theta)$ is differentiable everywhere. Also
\begin{equation}
|p_j(\theta)| \leq \binom{n}{j} \int_0^1 x(1-x)\,d\pi_{\theta}(x) = \binom{n}{j}\frac{2\theta}{1+2\theta}
\end{equation}
hence $p_j(\theta) \rightarrow 0$ as $\theta \downarrow 0$.   
\end{proof}

Now to establish the Poisson representation, it is enough to apply the well-known Poisson approximation to the multinomial distribution.  We use:

\begin{theorem} \cite{mcdonald1980}.  If $(Z_0,\ldots,Z_n)$ is multinomial with parameters $(L,p_0,\ldots,p_n)$, and $(V_1,\ldots,V_{n-1})$ are independent Poissons with means $Lp_j$, then
\begin{equation}\label{poiss3}
\|(Z_1,\ldots,Z_{n-1}) - (V_1,\ldots,V_{n-1})\| \leq 2L \left(\sum_{j=1}^{n-1} p_j\right)^2 
\end{equation}
where $\|\cdot \|$ is the total variation norm of measures.
\end{theorem}
The Poisson representation is now obvious, since $p_j(\theta) = O(\theta) = O(1/L)$ by Lemma \ref{baselemma}  and therefore the right-hand side of (\ref{poiss3}) goes to zero as $L \rightarrow \infty$.  Since $(V_1,\ldots,V_{n-1})$ are converging to a sequence of independent Poisson distributions with finite means $c_j \Theta$, where $c_j$ are as in Lemma \ref{baselemma}, so must  $(Z_1,\ldots,Z_{n-1})$.

\subsection{The average number of segregating sites.}

In this section, we study the number of segregating sites $S_n$ in the infinite-sites $\Lambda$ model, deriving the recursion (\ref{srecur}) for the average diversity measure $s_n = \mathbb{E}S_n/\theta$, and use it to  obtain asymptotic expressions for diversity.  

Let $\pi_\theta$ be the stationary distribution of the $\Lambda$ process under two way symmetric mutation $\theta$.  Define $H_n(\theta)$ as the heterozygosity measure
\begin{equation}
H_n(\theta) = \sum_{j=1}^{n-1} \int_0^1 \binom{n}{j} x^j (1-x)^{n-j}\, d\pi_\theta(x)
\end{equation}
Applying the binomial theorem, 
\begin{equation}\label{hn}
H_n(\theta) = \int_0^1 (1-x^n -(1-x)^n)\, d\pi_\theta(x) = \int_0^1 (1-2x^n) \, d\pi_\theta(x)
\end{equation}
where the second equality follows from symmetry of the stationary distribution.  Now, define the diversity measure $s_n \equiv \mathbb{E}S_n/\Theta$.  We have from (\ref{ww}):
\begin{equation}\label{ssn}
s_n = \lim_{\theta \downarrow 0} \frac{1}{\theta} H_n(\theta)
\end{equation}
Under symmetric mutation $\theta= \theta_1 = \theta_2$, the recursion formulae for moments (\ref{gt}) reads
\begin{equation}\label{pl}
m_n(\theta) = \frac{n\theta + \sum_{j=1}^{n-1} a_{jn} m_j(\theta)}{n\theta + a_{nn}}
\end{equation}
From symmetry of the stationary distribution and differentiability of $m_n$, there are numbers $\{v_n\}$ such that $m_n(\theta) = 1/2 + v_n \theta + O(\theta^2)$.  Inserting this into the right-hand side of (\ref{pl}), and expanding in a Taylor series, we obtain, by comparing the first-order coefficients,
a recursion for $v_n$:
\begin{equation} v_n = \frac{-n/4  + \sum_{j=1}^{n-1} a_{jn} v_j}{a_{nn}}
\end{equation}
The equations (\ref{hn}) and (\ref{ssn}) imply that $s_n = -2v_n$.   Thus the corresponding recursion for $\{s_n\}$ is
\begin{equation}\label{snr}
s_n = \frac{n/2 + \sum_{j=1}^{n-1} a_{jn} s_j}{a_{nn}}
\end{equation}
where we initialize $s_1 = 0$.
Observe that by the binomial theorem, one has the relation
\begin{equation}
a_{nn} = \sum_{j=1}^{n-1} a_{jn}
\end{equation}
Thus the numbers $a_{jn}/a_{nn}$ define a probability measure on the set $j \in \{1,\ldots,n-1\}$.  By studying this measure and the recurrence relation defining $s_n$, we may derive the asymptotics for $s_n$.  

Now suppose that the underlying $\Lambda$ process is associated with a $\Lambda$ measure with support bounded away from zero.  Then from (\ref{aj1}), $a_{nn} = \int_0^1 \lambda^{-2} d\Lambda + O(\gamma^n)$, for some $0 < \gamma < 1$.  Therefore, $ a_{nn}^{-1} = \left(\lambda^{-2} d\Lambda \right)^{-1} + O(\gamma^n)$.  Using this estimate in the recursion (\ref{ssn}), and defining $A = \left(\int_0^1 \lambda^{-2} d\Lambda\right)^{-1}$, 
\begin{equation}
s_n = An/2 + A \sum_{j=1}^{n-1} a_{jn} s_j + O(\gamma^n)
\end{equation}
Now, we shall prove
\begin{theorem}
\begin{equation}
s_n = Cn + O(\log n), \qquad n \rightarrow \infty
\end{equation}
where 
\begin{equation}
C = \frac{A}{2-2A\int_0^1 \lambda^{-2} (1-\lambda)\,d\Lambda}
\end{equation}
\end{theorem}
\begin{proof}
Let $s_n = Cn + g_n$.  We find, using the explicit expressions (\ref{aj2}),
\begin{equation}\label{gb}
Cn + g_n = An/2 + A \sum_{j=1}^{n-1} \int_0^1 \lambda^{-2} \binom{n}{j-1} \lambda^{n-(j-1)} (1-\lambda)^{j-1} \, d\Lambda \cdot (Cj + g_j) + O(\gamma^n)
\end{equation}
The right-hand side of the above is
\begin{equation}
An/2 + A \sum_{j=1}^{n+1} \int_0^1 Cj\lambda^{-2} \binom{n}{j-1} \lambda^{n-(j-1)} (1-\lambda)^{j-1} d\Lambda  + A \sum_{j=1}^{n-1} \int_0^1 g_j \lambda^{-2} \binom{n}{j-1} \lambda^{n-(j-1)} (1-\lambda)^{j-1} d\Lambda  + O(\gamma^n)
\end{equation}
where we have subtracted two binomial terms and exponentially bounded them.  Using the formula for the mean of a binomial random variable with parameters $(n, 1-\lambda)$ in the first series, this is equivalent to 
\begin{equation}
An/2 + An \int_0^1 C(1-\lambda)\lambda^{-2}d\Lambda +  A \sum_{j=1}^{n-1} \int_0^1 g_j \lambda^{-2} \binom{n}{j-1} \lambda^{n-(j-1)} (1-\lambda)^{j-1} d\Lambda+ O(1)
\end{equation}
Inserting this expression into (\ref{gb}), dividing by $n$, and using the definition of $C$, we find that $g_n$ satisfies the relation
\begin{equation}
g_n = A \sum_{j=1}^{n-1}\int_0^1  g_j \lambda^{-2} \binom{n}{j-1} \lambda^{n-(j-1)} (1-\lambda)^{j-1} d\Lambda+ O(1/n)
\end{equation}
Completing the series with two binomial terms once more, 
\begin{equation}\label{rr}
g_n = A\int_0^1  \lambda^{-2} (\mathbb{E}_{\lambda,n} g)\, d\Lambda + O(1/n)
\end{equation}
where the expectation occurs with respect to a binomial random variable with parameters $(n, 1-\lambda)$.  We may further interpret the integral term as an expectation of $g$ under a mixture of binomials with mixing weights $1/\lambda^2$. 

 We can now prove that that $g_n = O(\log n)$.  To be explicit in (\ref{rr}), let $B$ be a constant such that
 \begin{equation}
\left| g_n - A\int_0^1  \lambda^{-2} (\mathbb{E}_{\lambda,n} g)\, d\Lambda \right| \leq B/n 
 \end{equation}
 for all $n$.
We shall show, by induction, that  $g_n \leq B \log n $, for large $n$.  By enlargening $B$, the base case $g_0$ is satisfied.  Assume that the statement it is true up to $n$.  Then using Jensen's inequality on the mixture in (\ref{rr}),
\begin{equation}
g_{n+1} \leq B\log \left( A \int_0^1 \lambda^{-2} (1-\lambda) n \, d\Lambda \right) +B/n
\end{equation}
If $\epsilon>0$ is a number smaller than $\inf \text{supp } \Lambda$, then we can, keeping in mind the definition of $A$, further bound this expression  by:
\begin{align}
g_{n+1} &\leq B \log ((1-\epsilon) n) + B/n\\
&=B \log ((1-\epsilon)n) + B \log (n+1) - B\log (n) + O(1/n^2)\\
&\leq B\log(n+1)
\end{align}
Thus Theorem 2 is proved.
\end{proof}
The constant $C$ in Theorem 2 is algebraically equivalent to $C(\Lambda)$ in (\ref{asympform}).

\subsection{An Optimization Principle for the average number of segregating sites.}

In this section we prove the following theorem:
\begin{theorem} For every sample size $n$, and for fixed mutation rate $\Theta$, the minimum and maximum values of $\mathbb{E}S_n$ over the class of $\Lambda$ processes is achieved within the class of pure Eldon-Wakeley processes; that is, where $\Lambda = \delta_\lambda$, for $0 \leq \lambda \leq 1$.  
\end{theorem}
\begin{proof} It is evident from the recursion for $\mathbb{E}S_n$ in (\ref{snr}) that the average number of segregating sites must have the form:
\begin{equation}
\mathbb{E}S_n/\Theta = \frac{\int_0^1 f_n(\lambda)\, d\Lambda}{\int_0^1 g_n(\lambda)\, d\Lambda}
\end{equation}
for some functions (in fact, polynomials) $f_n, g_n$.  Because $\mathbb{E}S_n$ is positive for every pure Eldon-Wakeley process (where $\Lambda = \delta_\lambda$), we can without loss of generality assume that $f_n, g_n \geq 0$.  For such functions, we have the lemma:
\begin{lemma} For positive functions $f,g$ defined on $[0,1]$, we have the inequalities:
\begin{equation}
\min_\lambda \frac{f(\lambda)}{g(\lambda)} \leq \frac{\int_0^1 f(\lambda)\,d\Lambda}{\int_0^1 g(\lambda)\, d\Lambda} \leq \max_\lambda \frac{f(\lambda)}{g(\lambda)}
\end{equation}
\end{lemma}
\begin{proof} We prove the lower bound; the upper bound is established in the same way.  Say $\Lambda$ is a two-point measure: $\Lambda = p_1 \delta_{\lambda_1} + p_2 \delta_{\lambda_2}$.  Then elementary manipulations show
\begin{equation}
\frac{p_1 f(\lambda_1) + p_2 f(\lambda_2)}{p_1 g(\lambda_1) + p_2 g(\lambda_2)} \geq \text{min } \{f(\lambda_1)/g(\lambda_1), f(\lambda_2)/g(\lambda_2)\} 
\end{equation}
By induction one easily generalize to measures concentrated at any finite number of points.  Finally, the full case is obtained by taking weak limits of measures concentrated at a finite number of points.
\end{proof}

Returning to the proof of the main theorem, one sees that Lemma 2 immediately implies the result, since $f_n(\lambda)/g_n(\lambda)$ is precisely $\mathbb{E}S_n/\Theta$ for the case $\Lambda = \delta_\lambda$.  
\end{proof}

The optimization principle can be refined.  Let $\Lambda$ be any probability
measure whose support excludes a neighborhood of zero, and suppose that
$\lambda_{\text{min}}$ and $\lambda_{\text{max}}$ are the smallest and largest
values, respectively, on which $\Lambda$ is supported, i.e. $\lambda_{\text{min}}
= \inf  \text{supp } \Lambda$, and $\lambda_{\text{max}} = \sup \text{supp }
\Lambda$.  Then
\begin{equation}
\frac{\lambda_{\text{min}}}{2} \leq C(\Lambda) \leq \frac{\lambda_{\text{max}}}{2}
\end{equation}
In conjunction with (\ref{asympform}), this shows that the asymptotic growth rate
in $\mathbb{E}S_n$ for a $\Lambda$ -process whose drift measure is supported on
the interval $[\lambda_{min}, \lambda_{max}]$ can be lower and upper-bounded by
the rates of growth in $\mathbb{E}S_n$ of two  pure  processes, with
parameters $\lambda_{min}$ and $\lambda_{max}$, respectively.  In other words, the
effect of mixing any two pure $\Lambda$ processes always results  in a process
whose equilibrium diversity is intermediate relative to the diversities of the
pure models.

\section{Acknowledgments}
The authors are grateful to Warren Ewens and Charles Epstein for many fruitful discussions.

\bibliography{paper7}
\bibliographystyle{mychicago}

\end{document}